\documentclass[conference]{IEEEtran}
\IEEEoverridecommandlockouts
\usepackage{latexsym}
\usepackage{amsmath}
\usepackage{graphicx}
\usepackage{hyperref}
\usepackage{verbatim}
\usepackage{pgf}
\usepackage{tikz}
\usetikzlibrary{shapes, arrows,automata}
\usepackage[latin1]{inputenc}

\usepackage{amsthm}

\usepackage{caption}
\usepackage{subcaption}

\makeatletter
\newif\if@restonecol
\makeatother

\usepackage[ruled,noend,linesnumbered,lined]{algorithm2e}

\newtheorem{theorem}{Theorem}
\newtheorem{example}{Example}

\newtheorem{lemma}[theorem]{Lemma}

\usepackage{stmaryrd}

\newcommand{\G}{\Box}
\newcommand{\GF}{\Box \Diamond}
\newcommand{\X}{\bigcirc}
\newcommand{\F}{\Diamond}
\newcommand{\Tc}{\mathcal{T}_c}

\newcommand{\set}[1]{\left\{ #1 \right\}}

\newcommand{\seq}[1]{\langle #1 \rangle}

\usepackage{color}

\begin{document}

\title{Counter-Strategy Guided Refinement of GR(1) Temporal Logic Specifications}  
\author{\IEEEauthorblockN{Rajeev Alur, Salar Moarref, and Ufuk Topcu}
\IEEEauthorblockA{University of Pennsylvania, Philadelphia, USA. \tt{\{alur,moarref,utopcu\}@seas.upenn.edu}}
\thanks{This research was partially supported by NSF Expedition in
Computing project ExCAPE (grant CCF 1138996), and AFOSR (grant number FA9550-12-1-0302).}}

\maketitle

\begin{abstract}
The reactive synthesis problem is to find a finite-state controller that satisfies a given
temporal-logic specification regardless of how its environment behaves.
Developing a formal specification is a challenging and tedious task and initial specifications are
often unrealizable.  In many cases, the source of unrealizability is the lack of adequate assumptions on the
environment of the system.  In this paper, we consider the problem of automatically correcting an unrealizable
specification given in the generalized reactivity (1) fragment of linear temporal logic
by adding assumptions on the environment.  When a temporal-logic specification is unrealizable,
the synthesis algorithm computes a counter-strategy as a witness.
Our algorithm then analyzes this counter-strategy
and synthesizes a set of candidate environment assumptions that can
be used to remove the counter-strategy from the environment's possible behaviors.
We demonstrate the applicability of our approach with several case studies.
\end{abstract}

%

\section{Introduction}
Automatically synthesizing a system from a high-level specification is an ambitious goal in the design of reactive systems.
The synthesis problem is to find a system that satisfies the specification regardless of how its environment behaves. 
Therefore, it can be seen as a two-player game between the environment and the system.
The environment attempts to violate the specification while the system tries to satisfy it. 
A specification is \emph{unsatisfiable} if there is no input and output trace that satisfies the specification. 
A specification is \emph{unrealizable} if there is no system that can implement the specification. 
That is, the environment can behave in such a way that no matter how the system reacts, the specification would be violated. 
In this paper we consider specifications which are satisfiable but unrealizable. 
We address the problem of strengthening the constraints over the environment by adding assumptions in order to achieve realizability. 

Writing a correct and complete formal specification which conforms to the (informal) design intent is a hard and tedious task \cite{ChatterjeeHJ08, debugCS}. 
Initial specifications are often incomplete and unrealizable.
Unrealizability of the specification is often due to inadequate environment assumptions. 
In other words, assumptions about the environment  are too weak, 
leading to an environment with too many behaviors that make it impossible for the system to satisfy the specification. 
Usually there is only a rough and incomplete model of the environment in the design phase; thus it is easy to miss 
assumptions on the environment side. 
We would like to automatically find such \emph{missing} assumptions that can be added to the specification and make it realizable. 
Computed assumptions can be used to give the user insight into the specification. 
They also provide ways to correct the specification. 
In the context of compositional synthesis \cite{kupferman2006safraless,ozay2011distributed}, 
derived assumptions based on the components specifications can be used to construct interface rules between the components. 

An unrealizable specification cannot be executed or simulated which makes its debugging a challenging task.  
Counter-strategies are used to explain the reason for unrealizabilty of linear temporal logic (LTL) specifications \cite{debugCS}. 
Intuitively, a counter-strategy defines how the environment can react to the outputs of the system in order to enforce the system to violate the specification. 
Konighofer et al. in \cite{debugCS} show how such a counter-strategy can be computed for an unrealizable LTL specification. 
The requirement analysis tool RATSY \cite{Bloem_ratsy} implements their method for a fragment of LTL known as generalized reactivity (1) (GR(1)). 
We also consider GR(1) specifications in this paper because 
the realizability and synthesis problems for GR(1) specifications can be solved efficiently in polynomial time and
GR(1) is expressive enough to be used for interesting real-world problems \cite{bloem2012synthesis,wongpiromsarn2010receding}.

Counter-strategies can still be difficult to understand by the user especially for larger systems. 
We propose a debugging approach which uses the counter-strategies to strengthen the assumptions on the environment in order to make the specification realizable. 
For a given unrealizable specification, our algorithm analyzes the counter-strategy and synthesizes a set of \emph{candidate} assumptions in 
the GR(1) form (see section \ref{sec:prelims}). 
Any of the computed candidate assumptions, if added to the specification, restricts the environment in such a way that it cannot behave according to 
the counter-strategy---without violating its assumptions---anymore. 
Thus we say the counter-strategy is ruled out from the environment's possible behaviors by adding the candidate assumption to the specification.

The main flow for finding the missing environment assumptions is as follows.
If the specification is unrealizable, a counter-strategy is computed for it. A set of \emph{patterns} are then synthesized by processing an abstraction of the counter-strategy. 
Patterns are LTL formulas of special form that define the structure for the candidate assumptions. 
We ask the user to specify a set of variables to be used for generating candidates for each pattern. 
The user can specify the set of variables which she thinks contribute to unrealizability or are underspecified. 
The variables are used along with patterns to generate the candidate assumptions.  
Any of the synthesized assumptions can be added to the specification to rule out the counter-strategy. 
The user can choose an assumption from the candidates in an interactive way or our algorithm can automatically search for it. 
The chosen assumption is then added to the specification and the process is repeated with the new specification. 

The contributions of this paper are as follows:
We propose algorithms to synthesize environment assumptions by directly processing the counter-strategies. 
  We give a counter-strategy guided synthesis approach that finds the missing environment assumptions. 
The suggested refinement can be validated by the user to ensure compatibility with her design intent and can be added to the specification to make it realizable.
We demonstrate our approach with examples and case studies.

The problem of correcting an unrealizable LTL specification by constructing an additional environment assumption is studied by Chatterjee et al. 
in \cite{ChatterjeeHJ08}. 
 They give an algorithm for computing the assumption which only constrains the environment and is as weak as possible. 
Their approach is more general than ours as they consider general LTL specifications. 
However, the synthesized assumption is a B\"{u}chi automaton which might not translate to an LTL formula and can be difficult for the user to understand 
(for an example, see Fig. $3$ in \cite{ChatterjeeHJ08}). 
Moreover, the resulting specification is not necessarily compatible with the design intent \cite{LiDS11}. 
 Our approach generates a set of assumptions in GR(1) form that can easily be validated by the user and be used to make the specification realizable.

The closest work to ours is the work by Li et al. \cite{LiDS11} where they propose 
a template-based specification mining approach to find additional assumptions on the environment 
that can be used to rule out the counter-strategy. 
A template is an LTL formula with at least one placeholder, $?_b$, that can be instantiated by the Boolean variable $b$ or its negation. 
Templates are used to impose a particular structure on the form of generated candidates and are engineered by the user based on her knowledge of the environment. 
A set of candidate assumptions is generated by enumerating all possible instantiations of the defined templates. 
 For a given counter-strategy, their method finds an assumption from the set of candidate assumptions which is satisfied by the counter-strategy. 
By adding the negation of such an assumption to the specification, they remove the behavior described by the counter-strategy from the environment. 
Similar to their work, we consider unrealizable GR(1) specifications and achieve realizability by adding environment assumptions to the specification. 
But, unlike them, we directly work on the counter-strategies to synthesize a set of candidate assumptions that can be used to rule out the counter-strategy.  
Similar to templates, patterns impose structure on the assumptions. 
However, our method synthesizes the patterns based on the counter-strategy and the user does not need to manipulate them. 
We only require the user to specify  a subset of variables to be used in the search for the missing assumptions. 
The user can specify a subset that she thinks leads to the unrealizability. 
In our method, the maximum number of generated assumptions for a given counter-strategy is independent from what subset of variables is considered, 
whereas increasing the size of the chosen subset of variables in \cite{LiDS11} will result in exponential growth in the number of candidates, 
while only a small number of them might hold over all runs of the counter-strategy (unlike our method).
Moreover, we compute the weakest environment assumptions for the considered structure and given subset of variables. 
Our work takes an initial step toward bridging the gap between \cite{ChatterjeeHJ08} and \cite{LiDS11}. 
Our method synthesizes environment assumptions that are simple  formulas, making them easy to understand and practical, 
and they also constrain the environment as weakly as possible within their structure.
We refer the reader to \cite{LiDS11} for a survey of related work.     
 
  

\section{Preliminaries}
\label{sec:prelims}

%
%
%
%


Linear temporal logic (LTL) is a formal specification language with two kinds of operators: 
logical connectives (negation ($\neg$), disjunction ($\vee$), conjunction ($\wedge$) and 
implication ($\rightarrow$)) and temporal modal operators 
(next ($\bigcirc$), always ($\Box$), eventually ($\Diamond$)
 and until ($\mathcal{U}$)). 
 Given a set $P$ of atomic propositions, an LTL formula is defined inductively as follows: $1$) any atomic proposition $p \in P$ is an LTL formula. 
$2$)  if $\phi$ and $\psi$ are LTL formulas, then $\neg \phi$, $\phi \vee \psi$, $\bigcirc \phi$ and $\phi \,\mathcal{U}\, \psi$ are also LTL formulas.
Other operators can be defined using the following rules: 
$\phi \wedge \psi = \neg (\neg \phi \vee \neg \psi)$, $\phi \rightarrow \psi = 
\neg \phi \vee \psi$, $\Diamond \phi = \tt{True} \,\mathcal{U}\, \phi$ and $\Box \phi = \neg \Diamond \neg \phi$.
An LTL formula is interpreted over infinite words $\omega \in (2^P)^\omega$. 
 For an LTL formula $\phi$, 
we define its language $\mathcal{L}(\phi)$ to be the set of infinite words that satisfy $\phi$, i.e., $\mathcal{L}(\phi)=\set{\omega \in (2^P)^\omega~|~ \omega \models \phi}$.  
 
A finite transition system (FTS)  is a tuple $\mathcal{T}=\seq{Q,Q_0,\delta}$ where $Q$ is a finite set of states, $Q_0 \subseteq Q$ is 
 the set of initial states and $\delta \subseteq Q \times Q$ is the transition relation. 
 An \emph{execution} or \emph{run} of a FTS is an infinite sequence of states $\sigma = q_0q_1q_2...$ where $q_0 \in Q_0$ and for any $i\geq0$, $q_i \in Q$ and 
 $(q_i,q_{i+1}) \in \delta$.
The language of a FTS $\mathcal{T}$ is defined as the set $\mathcal{L}(\mathcal{T})=\set{\omega \in Q^\omega ~|~ \omega \text{ is a run of } \mathcal{T}}$, i.e., 
the set of (infinite) words generated by the runs of $\mathcal{T}$. We often consider a finite transition system as a directed graph with a natural bijection 
between the states and transitions of the FTS and vertices and edges of the graph, respectively. 
Formally for a FTS $\mathcal{T}=\seq{Q,Q_0,\delta}$, we define the graph $\mathcal{G_T}=\seq{V,E}$ where 
each $v_i \in V$ corresponds to a unique state  $q_i \in Q$, and $(v_i,v_j) \in E$  if and only if $(q_i,q_j) \in \delta$.  

 

Let $P$ be a set of atomic propositions, partitioned into input, $I,$ and output, $O,$ propositions. A {\it Moore transducer} is a tuple $M=(S, s_0, \mathcal{I}, \mathcal{O}, \delta, \gamma)$, 
where $S$ is the set of states, $s_0 \in S$ is the initial state, $\mathcal{I}=2^I$ is the input alphabet, $\mathcal{O}=2^O$ is the output alphabet, 
$\delta: S \times \mathcal{I} \rightarrow S$ is the transition function and $\gamma: S \rightarrow \mathcal{O}$ is the state output function. 
A {\it Mealy} transducer is similar, except that the state output function is $\gamma: S \times \mathcal{I} \rightarrow \mathcal{O}$.
For an infinite word $\omega \in \mathcal{I}^\omega$, a run of $M$ is the infinite sequence $\sigma \in S^\omega$ such that $\sigma_0 = s_0$ and 
for all $i \geq 0$ we have $\sigma_{i+1}=\delta(\sigma_i,\omega_i)$. 
The run $\sigma$ on input word $\omega$ produces an infinite word $M(\omega) \in (2^P)^\omega$ such that $M(\omega)_i = \gamma(\sigma_i) \cup \omega_i$
 for all $i \geq 0$. The language of $M$ is the set $\mathcal{L}(M)=\{M(\omega)~|~ \omega \in \mathcal{I}^\omega\}$ of infinite words generated by runs of 
$M$. 

An LTL formula $\phi$ is \emph{satisfiable} if there exists an infinite word $\omega \in (2^P)^\omega$ such that $\omega \models \phi$. 
A Moore (Mealy) transducer $M$ satisfies an LTL formula $\phi$, written as $M \models \phi$, if $\mathcal{L}(M) \subseteq \mathcal{L}(\phi)$. 
An LTL formula $\phi$ is \emph{Moore (Mealy) realizable} if there exists a Moore (Mealy, respectively) transducer $M$ such that $M \models \phi$. 
The \emph{realizability problem} asks whether there exists such a transducer for a given LTL specification $\phi$.

A \emph{two-player deterministic game graph} is a tuple $\mathcal{G}=(Q,Q_0,E)$ where $Q$ can be partitioned into two disjoint sets 
$Q_1$ and $Q_2$. $Q_1$ and $Q_2$ are the sets of states of player $1$ and $2$, respectively. 
$Q_0$ is the set of initial states. 
$E=Q \times Q$ is the set of directed edges. 
Players take turns to play the game. At each step, if the current state belongs to $Q_1$, player $1$ chooses the next state. 
Otherwise player $2$ makes a move. 
A \emph{play} of the game graph $\mathcal{G}$ is an infinite sequence $\sigma = q_0q_1q_2...$ of states such that $q_0 \in Q_0$, and $(q_i,q_{i+1}) \in E$ for all $i \geq 0$. 
We denote the set of all plays by $\Pi$.
A \emph{strategy} for player $i \in \{1,2\}$ is a function $\alpha_i:Q^*. Q_i \rightarrow Q$ 
that chooses the next state given a finite sequence of states which ends at a player $i$ state. 
A strategy is \emph{memoryless} if it is a function of current state of the play, i.e., $\alpha_i: Q_i \rightarrow Q$. 
Given strategies $\alpha_1$ and $\alpha_2$ for players and a state $q \in Q$, the \emph{outcome} is the play 
starting at $q$, and evolved according to $\alpha_1$ and $\alpha_2$. 
Formally, $outcome(q,\alpha_1,\alpha_2)=q_0q_1q_2...$ where $q_0 = q$, and for all $i \geq 0$ we have  $q_{i+1}=\alpha_1(q_0q_1...q_i)$ if $q_i \in Q_1$ and $q_{i+1}=\alpha_2(q_0q_1...q_i)$ if $q_i \in Q_2$. 
An \emph{objective} for a player is a set $\Phi \subseteq \Pi$ of plays. 
A strategy $\alpha_1$ for player $1$ is winning for some state $q$ if for every strategy $\alpha_2$ of player $2$, 
we have $outcome(q, \alpha_1, \alpha_2) \in \Phi$. 

Given an LTL formula $\phi$ over $P$ and a partitioning of $P$ into $I$ and $O$, 
the \emph{synthesis problem} is to find a Mealy transducer $M$ with input alphabet $\mathcal{I}=2^I$ and output alphabet $\mathcal{O}=2^O$ that satisfies $\phi$. 
This problem can be reduced to computing winning strategies in game graphs. 
A deterministic game graph $G$, and an objective $\Phi$ can be constructed such that $\phi$ is realizable if and only if the system (player $1$) has a memoryless winning strategy from the initial state in $G$ \cite{pnueli1989synthesis}. 
Every memoryless winning strategy of the system can be represented by a Mealy transducer that satisfies $\phi$.
If the specification $\phi$ is unrealizable, then the environment (player $2$) has a winning strategy.
A \emph{counter-strategy} for the synthesis problem is a strategy for the environment that can falsify the specification, no matter how the system 
plays. Formally, a counter-strategy can be represented by a Moore transducer $M_c = (S',s_0',\mathcal{I}',\mathcal{O}', \delta', \gamma')$ that satisfies $\neg \phi$, 
where $\mathcal{I}'=\mathcal{O}$ and $\mathcal{O}'=\mathcal{I}$ are the input and output alphabet for $M_c$ which are generated by the system and the environment, respectively.

 In this paper, we consider specifications of the form 
\begin{equation}
\label{gr1}
  \phi = \phi_e \rightarrow \phi_s, 
\end{equation}
where $\phi_{\alpha}$ for $\alpha \in \{e,s\}$ can be written as a conjunction of the following parts:

\begin{itemize}
\item $\phi_i^{\alpha}$: A Boolean formula over $I$ if $\alpha=e$ and over $I \cup O$ otherwise, characterizing the initial state. 
\item $\phi_t^{\alpha}$: An LTL formula of the form $\bigwedge_i \Box \psi_i$. 
Each subformula $\G \psi_i$ is either characterizing an invariant, in which case $\psi_i$ is a Boolean formula over $I \cup O$, or it is characterizing a transition relation, 
in which case $\psi_i$ is a Boolean formula over expressions $v$ and $\bigcirc v'$ where $v \in I \cup O$ and, $v' \in I$ if $\alpha=e$ and $v' \in I \cup O$ if $\alpha = s$.
\item $\phi_g^{\alpha}$: A formula of the form $\bigwedge_i  \Box \Diamond B_i$ characterizing fairness/liveness, 
where each $B_i$ is a Boolean formula over $I \cup O$.
\end{itemize}

For the specifications of the form in \eqref{gr1}, known as GR(1) formulas, Piterman et al. \cite{pitermangr1}   
show that the synthesis problem can be solved in polynomial time. 
Intuitively, in \eqref{gr1}, $\phi_e$ characterizes the assumptions on the environment and $\phi_s$ characterizes the correct behavior (guarantees) of the system. 
Any correct implementation of the specification guarantees to satisfy $\phi_s$, provided that the environment satisfies $\phi_e$. 

For a given unrealizable specification $\phi_e \rightarrow \phi_s$, 
we define a \emph{refinement} $\psi = \bigwedge_i \psi_i$ as a conjunction of a collection of environment assumptions $\psi_i$ in the GR(1) form such that 
$\phi_e \wedge \psi \rightarrow \phi_s$ is realizable. 
Intuitively it means that adding the assumptions $\psi_i$ to the specification results in a new specification which is realizable. 
We say a refinement $\psi$ is consistent with the specification $\phi_e \rightarrow \phi_s$ if $\phi_e \wedge \psi$ is satisfiable. 
Note that if $\phi_e \wedge \psi$ is not satisfiable, i.e., $\phi_e \wedge \psi = \tt{False}$, the specification $\phi_e \wedge \psi \rightarrow \phi_s$ is 
trivially realizable \cite{LiDS11}, but obviously $\psi$ is not an interesting refinement. 


\section{Problem Statement and Overview}
\label{PSOM}
\subsection{Problem Statement} Given a specification $\phi = \phi_e \rightarrow \phi_s$ in the GR(1) form which is satisfiable but unrealizable, 
find a  refinement $\psi = \bigwedge_i \psi_i$ as a conjunction of environment assumptions $\psi_i$ 
such that $\phi_e \wedge \psi$ is satisfiable 
and $\phi_e \wedge \psi \rightarrow \phi_s$ is realizable.

\subsection{Overview of the Method}
\label{overview}
We now give a high-level view of our method. Specification refinements are constructed in two phases.  
First, given a counter-strategy's Moore machine $M_c$, we build an abstraction which is a finite transition system $\Tc$. 
The abstraction preserves the structure of the counter-strategy (its states and transitions) while removing the input and output details.
The algorithm processes $\Tc$ and synthesizes a set of LTL formulas in special forms, called \emph{patterns}, which hold over \emph{all} runs of $\Tc$.  
Our algorithm then uses these patterns along with a subset of variables specified by the user to generate a set of LTL 
formulas which hold over \emph{all} runs of $M_c$. 
We ask the user to specify a subset of variables which she thinks contribute to the unrealizability of the specification. 
This set can also be used to guide the algorithm to generate formulas over the set of variables which are underspecified. 
Using a smaller subset of variables leads to simpler formulas that are easier for the user to understand.

The complement of the generated formulas form the set of candidate assumptions that can be used to rule out the counter-strategy from the environment's possible behaviors. 
We remove the candidates which are not consistent with the specification in order to avoid a trivial solution $\tt{False}$. 

Any assumption from the set of generated candidates can be used to rule out the counter-strategy. 
Our approach does a breadth-first search over the candidates. If adding any of the candidates makes the specification realizable, 
the algorithm returns that candidate as a solution. 
Otherwise at each iteration, the process is repeated for any of the new specifications resulting from 
adding a candidate. 
The depth of the search is controlled by the user. 
The search continues until either a consistent refinement is found or the algorithm cannot find one within the specified depth (hence the search algorithm is sound, but not complete). 


\begin{figure}
\centering
\begin{subfigure}{.5\linewidth}
  \centering
  \includegraphics[width=\textwidth , scale=0.07]{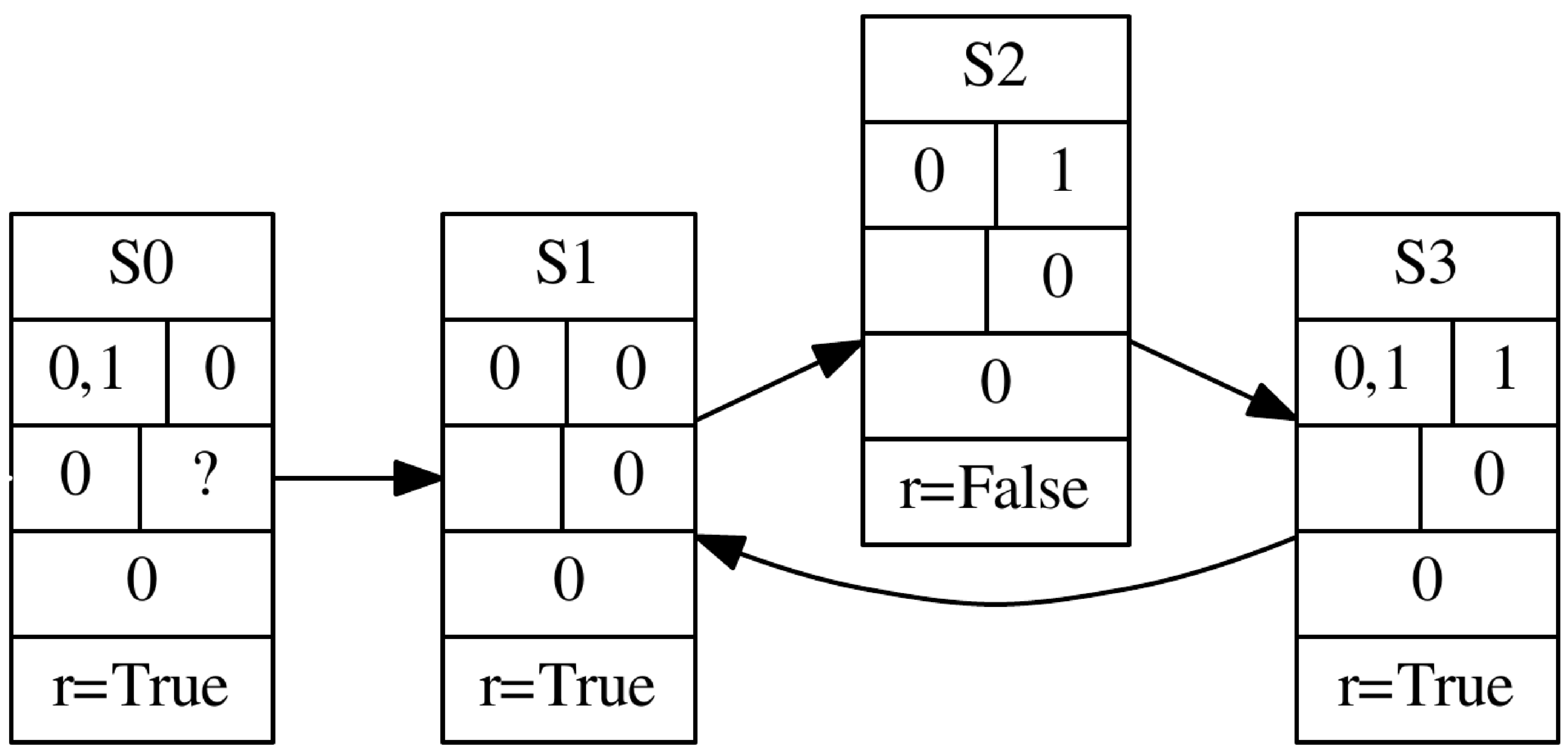}
  \caption{}
\end{subfigure}%
\hspace{0.5cm}
\begin{subfigure}{.4\linewidth}
  \centering
  \begin{tikzpicture}[->,>=stealth',shorten >=1pt,auto,node distance=2cm, semithick]
  \tikzstyle{every state}=[text=black, scale=0.7]

  \node[state] (A)                    {$q_0$};
  \node[state]         (B) [right of=A] {$q_1$};
  \node[state]         (C) [above right of=B] {$q_2$};
  \node[state]	       (D) [below right of=C] {$q_3$};
  \node at (0,-0.6) {};
  
  \path (A) edge  node {} (B)
        (B) edge  node {} (C)
	(C) edge  node {} (D)
	(D) edge  node {} (B);
  \end{tikzpicture}
  \caption{}
\end{subfigure}
\caption{(a) A counter-strategy produced by RATSY for the specification of Example \ref{ex1} with the additional assumption $\GF(\neg r)$. $c=\tt{True}$ is constant in all states.
(b) The abstract finite transition system for the counter-strategy of part (a).} \vspace{-0.15in}
\label{counter-strategy-example2}
\end{figure}
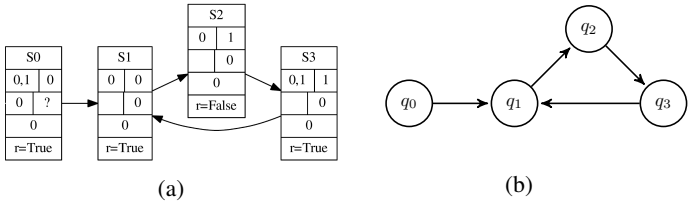

\begin{example}
\label{ex1}
 Consider the following example borrowed from \cite{LiDS11} with the environment variables $I=\set{r,c}$ and system variables $O=\set{g,v}$. 
Here $r,c,g$ and $v$ stand for \emph{request}, \emph{clear}, \emph{grant} and \emph{valid} signals respectively.
We start with no assumption, that is we only assume $\phi_e=\tt{True}$. Consider the following system guarantees: $\phi_1= \G(r \rightarrow \X\F g)$, 
$\phi_2= \G((c \vee g) \rightarrow \X \neg g)$, $\phi_3= \G(c \rightarrow \neg v)$ and $\phi_4= \G\F (g \wedge v)$.
Let $\phi_s$ be the conjunction of these formulas.
$\phi_1$ requires that every request must be granted eventually starting from the next step by setting signal $g$ to high. 
$\phi_2$ says that if clear or grant signal is high, then grant must be low at the next step. 
$\phi_3$ says if clear  is high, then the valid signal must be low. 
Finally, $\phi_4$ says that system must issue a valid grant infinitely often. 

The specification $\phi_e \rightarrow \phi_s$ is unrealizable. 
A simple counter-strategy is for the environment to keep $r$ and $c$ high at all times. 
Then, by $\phi_3$, $v$ needs to be always low and thus $\phi_4$ cannot be satisfied by any system. 
RATSY produces this counter-strategy 
which is then fed to our algorithm. 
An example candidate found by our algorithm to rule out this counter-strategy is the assumption $\psi=\GF(\neg r)$. 
Adding $\psi$ to the specification prevents the environment from always keeping $r$ high, thus the environment cannot use the counter-strategy anymore. 
However, the specification $\phi_e \wedge \psi \rightarrow \phi_s$ is still unrealizable. 
RATSY produces the counter-strategy shown in Figure \ref{counter-strategy-example2}(a) for the new specification. 
The new counter-strategy keeps the $c$ high all the times. 
The value of $r$ is changed depending on the state of the counter-strategy as shown in Figure \ref{counter-strategy-example2}(a). 
The top block in each state of Figure \ref{counter-strategy-example2}(a) is the name of the state. 
RATSY produces additional information, shown in middle blocks, on how the counter-strategy enforces the system to violate the specification. 
We do not use this information in the current version of the algorithm. 

The following formulas are examples of consistent refinements produced by our algorithm for the specification $\phi_e \rightarrow \phi_s$: 

\begin{itemize}
 \item  $\psi_1=\G(\neg r \vee \neg c) \wedge \G(r \vee \neg c)$
  \item $\psi_2=\G(r \rightarrow \X \neg c) \wedge \G(\neg r \rightarrow \X \neg c)$
  \item $\psi_3=\GF(\neg r) \wedge \G(\neg c \vee r) \wedge \G(\neg r \rightarrow \X \neg c))$
\end{itemize}

Assumptions in both of the refinements $\psi_1$ and $\psi_2$  imply $\G(\neg c)$, that is, adding them requires the environment to keep the signal $c$ always low. 
Although adding these assumptions make the specification realizable, it may not conform to the design intent. 
Refinement $\psi_3$ does not restrict $c$ like $\psi_1$ and $\psi_2$, and only assumes that the environment sets the signal $r$ to low infinitely often and 
that, when the request signal is low, the clear signal should be low at the same and the next step.  

\end{example}

\section{Specification Refinement}
\label{sec:method}
Algorithm \ref{findAssumptions} finds environment assumptions that can be added to the specification to make it realizable. 
It gets as input the initial unrealizable specification $\phi = \phi_e \rightarrow \phi_s$, the set $P$ of subsets of variables 
to be used in generated assumptions and the maximum depth $\alpha$ of the search. 
It outputs a consistent refinement $\psi$,  if it can find one within the specified depth. 

For an unrealizable specification, a counter-strategy is computed as a Moore transducer using the techniques in \cite{debugCS, Bloem_ratsy}. 
The counter-strategy is then fed to the \textbf{GeneratePatterns} procedure which constructs a set of patterns 
and is detailed in Section \ref{patterns-synthesis}. 
Procedure \textbf{GenerateCandidates}, described in Section \ref{sec:generateCandidates}, 
 produces a set of candidate assumptions in the form of GR(1) formulas using patterns and the set $P$ of variables. 
Algorithm \ref{findAssumptions} runs a breadth-first search to find a consistent refinement. 
Each node of the search tree is a generated candidate assumption, while the root of the tree corresponds to the assumption $\tt{True}$ (i.e., no assumption). 
Each path of the search tree starting from the root corresponds to a candidate refinement as conjunction of candidate assumptions of the nodes visited along the path. 
When a node is visited during the search, its corresponding candidate refinement is added to the specification. 
If the new specification is consistent and realizable, the refinement is returned by the algorithm. 
Otherwise, if the depth of the current node is less than the maximum specified, a set of candidate assumptions are generated based on the counter-strategy for the new specification and the search tree expands. 

In Algorithm \ref{findAssumptions}, the queue {\it CandidatesQ} keeps the candidate refinements  which are found during the search. 
At each iteration, a candidate refinement $\psi$ is removed from the head of the queue.
The procedure \textbf{Consistent} checks if $\psi$ is consistent with the specification $\phi$. 
If it is, the algorithm checks the realizability of the new specification $\phi_{new}=\phi_e \wedge \psi \rightarrow \phi_s$ using the procedure 
\textbf{Realizable} \cite{bloem2012synthesis,Bloem_ratsy}. 
If $\phi_{new}$ is realizable, $\psi$ is returned as a suggested refinement. 
Otherwise,  if the depth of the search for reaching the candidate refinement $\psi$ is less than $\alpha$, 
a new set of candidate assumptions are generated using the counter-strategy computed for $\phi_{new}$. 
Algorithm \ref{findAssumptions} keeps track of the number of counter-strategies produced along the path to reach a candidate refinement in order to compute its depth (\textbf{Depth($\psi$)}). 
Each new candidate assumption $\psi_{new}$ results in a new candidate refinement $\psi \wedge \psi_{new}$ which is added to the end of the queue for future processing .  
The algorithm terminates when either a consistent refinement $\psi$ is found, or there is no more candidates in the queue to be processed.  

\begin{algorithm}[t]
  \KwIn{$\phi = \phi_e \rightarrow \phi_s$, initial specification}
  \KwIn{$P$, set of subsets of variables to be used in patterns}
  \KwIn{$\alpha$, maximum depth of the search }
  \KwOut{$\psi$, additional assumptions such that $\phi_e \wedge \psi \rightarrow \phi_s$ is realizable}
  
  $M_c :=$ \textbf{CounterStrategy}($\phi$)\;
  Patterns := \textbf{GeneratePatterns}($M_c$)\;
  CandidatesQ := \textbf{GenerateCandidates}(Patterns,$P$)\;
  \While{CandidatesQ is not \textbf{Empty}}{
  	 $\psi$ :=  CandidatesQ.\textbf{DeQueue}\;
      \If{\textbf{Consistent}$(\phi$,$\psi)$}{
      	$\phi_{new} = \phi_e \wedge \psi \rightarrow \phi_s$\;
      	\If{\textbf{Realizable}($\phi_{new}$)} {return $\psi$\;}
      	\Else{
		\If{\textbf{Depth}$(\psi ) < \alpha$}{
			$M_c :=$ \textbf{CounterStrategy}($\phi_{new}$)\;
  			Patterns := \textbf{GeneratePatterns}($M_c$)\;
  			newCandidates := \textbf{GenerateCandidates}(Patterns,$P$) \;
			\ForEach{$\psi_{new} \in$ newCandidates}{
				CandidatesQ.\textbf{EnQueue($\psi \wedge \psi_{new}$)}\;
			}
	         }
	}
      }
   }
  return \textbf{No refinement was found}\;
  \caption{Specification Refinement}
  \label{findAssumptions} 
\end{algorithm}

\vspace{-0.05in}

\subsection{Generating Candidates}
\label{sec:generateCandidates}

Consider the Moore transducer $M_c = (S,s_0,\mathcal{I},\mathcal{O}, \delta, \gamma)$ of a counter-strategy, where $\mathcal{I}=2^O$ and $\mathcal{O}=2^I$, 
and $O$ and $I$ are the set of the system and environment variables, respectively. 
Given $M_c$, we construct a finite transition system $\Tc=\seq{Q,\set{q_0},\delta}$ which preserves the structure of 
the $M_c$ while removing all details about its input and output. More formally, for each state $s_i \in S$, $\Tc$ has a corresponding state $q_i \in Q$, and $q_0 \in Q$ is the state 
corresponding to $s_0 \in S$. There exists a transition $(q_i,q_j) \in \delta$ if and only if there exists $y \in \mathcal{I}$ such that $\delta(s_i,y)=s_j$. 
It is easy to see that any run of $\Tc$ corresponds to a run of $M_c$ and vice versa. 

By processing the abstract FTS $\Tc$ of the counter-strategy, we synthesize a set of patterns which are
LTL formulas of the form $\Diamond \Box \psi_1$, $\Diamond \psi_2$ and $\Diamond (\psi_3 \wedge \X \psi_4)$ that hold over all runs of $\Tc$. 
Each $\psi_i$ for $i \in \set{1,2,3,4}$ is a disjunction of a subset of states of $\Tc$, 
i.e., $\psi_i = \bigvee_{q \in Q_i} q$ where $Q_{i} \subseteq Q$. 
The complements of these formulas, $\Box \Diamond \neg \psi_1$ (liveness), $\Box \neg \psi_2$ (safety), and $\Box (\psi_3 \rightarrow \X \neg \psi_4)$ (transition), respectively,
 are of the desired GR(1) form and provide the structure for the candidate assumptions that can be used to rule out the counter-strategy. 
Note that similar to \cite{LiDS11}, we do not synthesize assumptions characterizing the initial state because they are easy to specify in practice. 
Besides, it is simple to discover them from the counter-strategy. 
Patterns are generated using simple graph search algorithms explained in Section \ref{patterns-synthesis}.

\begin{example}
\label{Patterns}
 Figure \ref{counter-strategy-example2}(b) shows the abstract FTS for the counter-strategy of Figure \ref{counter-strategy-example2}(a). 
 For this FTS our algorithm produces the set of patterns $\Diamond \Box(q_1 \vee q_2 \vee q_3)$, $\Diamond q_0, \Diamond q_1, \Diamond q_2, \Diamond q_3$, and 
$\Diamond( q_0 \wedge \X q_1),$ $\Diamond( q_1 \wedge \X q_2 ), \Diamond( q_2 \wedge \X q_3), \Diamond( q_3 \wedge \X q_1)$.
Any run of $\Tc$ satisfies all of the above formulas. 
For example $\Tc \models \F q_i$ for $i \in \set{0,1,2,3}$, meaning that any run of the $\Tc$ will eventually visit state $q_i$. 
The formula $\F (q_1 \wedge \X q_2)$ means that any run of $\Tc$ will eventually visit state $q_1$ and then state $q_2$ at the next step.
Also any run of $\Tc$ satisfies $\F\G (q_1 \vee q_2 \vee q_3)$, meaning that any run of $\Tc$ will eventually reach and stay in the set of states $\set{q_1,q_2,q_3}$. 
\end{example}
 
As we mentioned previously, each state $q_i \in Q$ of the FTS $\Tc$ corresponds to a state $s_i \in S$ of the Moore transducer $M_c$ of the counter-strategy. 
Also recall that each run of $\Tc$ corresponds to a run of $M_c$. 
$M_c$, at any state $s_i \in S$, outputs the propositional formula $\mathcal{V}_{s_i}=\gamma(s_i)$ which is a valuation over all environment variables. 
Formally, for any state $s_i \in S$ of $M_c$, we have 
$\mathcal{V}_{s_i} = \ell_1^i \wedge \ell_2^i \wedge ... \wedge \ell_n^i$ where each $\ell_j^i$ is a literal over the environment variable $x_j \in I$. 
We call $\mathcal{V}_{s_i}$ the state predicate of $s_i$ and also $q_i$. 
We replace the states in the patterns with their corresponding state predicates to get a set of formulas which 
hold over all runs of the counter-strategy. 

\begin{example}
 Consider the counter-strategy shown in Figure \ref{counter-strategy-example2}(a). 
The state predicates are $\mathcal{V}_{S0} = \mathcal{V}_{S1} = \mathcal{V}_{S3} = c \wedge r$ and $\mathcal{V}_{S2} = c \wedge \neg r$, 
where $S0,S1,S2$ and $S3$ are the states of $M_c$.
Using the patterns obtained in Example \ref{Patterns} and replacing the states with their corresponding state predicates,
we obtain LTL formulas which hold over all runs of $M_c$. 
For example, the pattern $\Diamond \Box(q_1 \vee q_2 \vee q_3)$ gives us the formula $\Diamond \Box((c \wedge r) \vee (c \wedge \neg r)) = \Diamond \Box c$. 
Replacing $q_2$ with $\mathcal{V}_{S2}$ in the pattern $\Diamond q_2$ leads to $\Diamond (c \wedge \neg r)$. 
Similarly, the pattern $\Diamond (q_1 \wedge \X q_2)$ gives $\Diamond ((c \wedge r) \wedge \X (c \wedge \neg r))$.
\end{example}

The structure of the state predicates and patterns is such that any subset of the environment variables can be used along with the patterns to generate candidates 
and the resulting formulas still hold over all runs of the counter-strategy. 
Algorithm \ref{findAssumptions} gets the set $P=\set{P_1,P_2,P_3,P_4}$ as input, where each $P_i$ is a subset of environment variables that should be used 
in the corresponding $\psi_i$ for generating the candidate assumptions from the patterns of the form $\Diamond \Box \psi_1$, $\Diamond \psi_2$ and $\Diamond (\psi_3 \wedge \X \psi_4)$.

\begin{example}
\label{subset}
 Assume that the designer specifies $P_1=\set{r}$, $P_2=\set{c}$, $P_3=\set{r,c}$ and $P_4=\set{c}$. 
Then the pattern $\Diamond \Box(q_1 \vee q_2 \vee q_3)$ results in  $\Diamond \Box( r \vee  \neg r \vee r) = \Diamond \Box \tt{True}$. 
From $\Diamond q_2$ we obtain $\Diamond c$, and $\Diamond (q_1 \wedge \X q_2)$ leads to $\Diamond ((c \wedge r) \wedge \X c)$.
Note that using a smaller subset of variables leads to simpler formulas (and sometimes trivial as in $\Diamond \Box (\tt{True})$). 
However, this simplicity may result in assumptions which put more constraints on the environment as we will show later.
\end{example}

The complement of the generated formulas form the set of candidate assumptions that can be used to rule out the counter-strategy.   
For instance, formulas $\GF (\neg r \wedge r) = \GF(\tt{False})$, $\G(\neg c)$, $\G((c \wedge r) \rightarrow \X(\neg c))$ and $\G((c \wedge \neg r) \rightarrow \X (\neg c))$ 
are the candidate assumptions computed based on the user input in Example \ref{subset}. 
Note that there might be repetitive formulas among the generated candidates. 
We remove the repeated formulas in order to prevent the process from checking the same assumption repeatedly. 
We also use some techniques to simplify the synthesized assumptions (see the Appendix).


%


\subsection{Removing the Restrictive Formulas}
Given two non-equivalent formulas $\phi_1$ and $\phi_2$ we say $\phi_1$ is \emph{stronger} than $\phi_2$ if $\phi_1 \rightarrow \phi_2$ holds. 
Assume $\psi_1$ and $\psi_2$ are two formulas that hold over all runs of the counter-strategy computed for the specification $\phi_e \rightarrow \phi_s$, 
and that $\psi_1 \rightarrow \psi_2$. 
Note that $\neg \psi_2 \rightarrow \neg \psi_1$ also holds, that is $\neg \psi_1$ is a \emph{weaker} assumption compared to $\neg \psi_2$. 
Adding either $\neg \psi_1$ or $\neg \psi_2$ to the environment assumptions $\phi_e$ rules out the counter-strategy. 
However, adding the stronger assumption $\neg \psi_2$ restricts the environment more than adding $\neg \psi_1$. 
That is, $\phi_e \wedge \neg \psi_2$ puts more constraints on the environment compared to $\phi_e \wedge \neg \psi_1$. 

As an example, consider the counter-strategy $M_c$ shown in Figure \ref{counter-strategy-example2}(a). 
Both $\psi_1=\F(c \wedge \neg r)$ and $\psi_2=\F(c)$  hold over all runs of $M_c$. 
Moreover, $\psi_1 \rightarrow \psi_2$. 
Consider the corresponding assumptions $\neg \psi_1 = \G(\neg c \vee r)$ and $\neg \psi_2 = \G(\neg c)$. 
Adding $\neg \psi_2$ restricts the environment more than adding $\neg \psi_1$. 
$\neg \psi_2$ requires the environment to keep the signal $c$ always low, whereas in case of $\neg \psi_1$, 
the environment is free to assign additional values to its variables. 
It only prevents the environment from setting $c$ to high and $r$ to low at the same time.  

We construct patterns which are strongest formulas of their specified form that hold over all runs of the counter-strategy. 
Therefore the generated candidate assumptions are the weakest formulas that can be constructed for the given structure and the user specified subset of variables. 


\subsection{Synthesizing Patterns}
\label{patterns-synthesis}

In this section we show how certain types of patterns can be synthesized using the abstract FTS $\Tc$ of the counter-strategy.
A pattern $\mathcal{P}$, is an LTL formula $\phi_{\mathcal{P}}$ which holds over all runs of the FTS $\Tc$, i.e., $\Tc \models \phi_{\mathcal{P}}$. 
We are interested in patterns of the form $\Diamond \Box \psi$, $\Diamond \psi$ and $\Diamond (\psi_1 \wedge \bigcirc \psi_2)$. 
The complements of these patterns are of the GR(1) form and, after replacing states with their corresponding state predicates, will yield to candidate assumptions 
for removing the counter-strategy. 

\subsubsection{Patterns of the Form $\Diamond \psi$}

For a FTS $\Tc = \seq{Q,\set{q_0}, \delta}$, we define a \emph{configuration} $C \subseteq Q$ as a subset of states of $\Tc$. 
We say a configuration $C$ is an \emph{eventually configuration} if for any run $\sigma$ of $\Tc$ there exists a state $q \in C$ and a time step $i \geq 0$ such that $\sigma_i = q$. 
That is, any run of $\Tc$ eventually visits a state from configuration $C$. 
It follows that if $C$ is an eventually configuration for $\Tc$, then $\Tc \models \F\bigvee_{q \in C} q$.
We say an eventually configuration $C$ is \emph{minimal} if there exists no $C' \subset C$ such that $C'$ is an eventually configuration. 
Note that removing any state $q \in C$ from a minimal eventually configuration leads to a configuration which is not an eventually configuration. 

Algorithm \ref{newEventuallyAlg} constructs eventually patterns which correspond to the minimal eventually 
configurations of $\Tc$ with size less than or equal to $\beta$. 
The larger configurations lead to larger formulas which are hard for the user to parse. 
The user can specify the value of $\beta$. Heuristics can also be used to automatically set $\beta$ based on the properties of $\Tc$, 
e.g. the maximum outdegree of the vertices in the corresponding directed graph $\mathcal{G}_{\mathcal{T}_c}$, 
where the outdegree of a vertex is the number of its outgoing edges.
In Algorithm \ref{newEventuallyAlg}, the set $\F\text{Configurations}$ keeps the minimal eventually configurations discovered so far. 
Algorithm \ref{newEventuallyAlg} initializes the sets Patterns and $\F\text{Configurations}$ to $\set{\F q_0}$ and $\set{q_0}$, respectively. 
Note that $\F q_0$ holds over all runs of $\Tc$. 
The algorithm then checks each possible configuration $Q' \subseteq Q-\set{q_0}$ with size less than or equal to $\beta$ 
in a non-decreasing order of $|Q'|$ to find minimal eventually configurations. 
Without loss of generality we assume that all states in $\Tc$ have outgoing edges\footnote{A transition from any state with no outgoing transition can be added to a dummy state with a self loop. Patterns which include the dummy state will be removed.}. 
At each iteration, a configuration $Q'$ is chosen. 
Algorithm \ref{newEventuallyAlg} checks if there is a minimal eventually configuration $Q''$ which is already discovered and $Q'' \subset Q'$. 
If such  $Q''$ exists, $Q'$ is not minimal. 
Otherwise, the algorithm checks if it is an eventually configuration 
by first removing all the states in $Q'$ and their corresponding incoming and outgoing transitions from $\Tc$ to obtain another FTS $\Tc'$. 
 Now, if there is an infinite run from $q_0$ in $\Tc'$, then there is a run in $\Tc$ that does not visit any state in $Q'$. 
Otherwise, $Q'$ is a minimal eventually configuration and is added to 
 $\F\text{Configurations}$. The corresponding formula $\psi$ = $\F \bigvee_{q \in Q'} q$ is also added to the set of eventually patterns. 
 Note that checking if there exists an infinite run in $\Tc'$ can be done by considering $\Tc'$ as a graph and checking if there is a reachable cycle from $q_0$, 
which can be done in linear time in number of states and transitions of $\Tc$. 
Therefore, the algorithm is of complexity $O(|Q|^\beta(|Q|+|\delta|))$. 

\begin{algorithm}[t]
  \KwIn{Finite state transition system $\Tc=\seq{Q,\set{q_0},\delta}$}
  \KwIn{$\beta$, maximum number of states in generated patterns}
  \KwOut{a set of patterns of the form $\Diamond \psi$ where $\Tc \models \Diamond \psi$}
  Patterns := $\set{\F q_0}$\;
  $\F \text{Configurations}$ := $\set{q_0}$\;
  \ForEach{$Q' \subseteq Q - \set{q_0}$ with non-decreasing order of $|Q'|$ where $|Q'| \leq \beta$}{
  	\If{$\not \exists Q'' \in \F \text{Configurations}$ s.t. $Q'' \subseteq Q'$}{
		Let $\Tc' = \seq{Q-Q', \set{q_0}, \delta'}$ where $\delta' = \set{(q,q') \in \delta | q \not \in Q' \wedge q' \not \in Q'}$\;
		\If{there is no infinite run from $q_0$ in $\Tc'$}{
			Add $Q'$ to $\F$Configurations\; 
			Let $\psi$ = $\F \bigvee_{q_i \in Q'} q_i$\;
			Add $\psi$ to Patterns\;
		}
	}
  }
  return Patterns\;
  \caption{Generating $\Diamond \psi$ patterns}
  \label{newEventuallyAlg}
\end{algorithm}

\begin{example}
\label{eventuallyPatternEx}
Consider the FTS shown in Figure \ref{exampleGraph}. Algorithm \ref{newEventuallyAlg} starts at initial configuration $\set{q_0}$ and 
generates the formula $\Diamond q_0$. None of $\set{q_1}$, $\set{q_2}$ or $\set{q_3}$ is an eventually configuration. 
For example for configuration $\set{q_1}$, there exists the run $\sigma = q_0, (q_3)^\omega$ which never visits $q_1$. 
Configurations $\set{q_1,q_3}$ and $\set{q_2,q_3}$ are minimal eventually configurations. 
For example removing $\set{q_1,q_3}$ will lead to a FTS with no infinite run (no cycle is reachable from $q_0$ in the corresponding graph). 
It is easy to see that configuration $\set{q_1,q_2}$ is not an eventually configuration. 
Configuration $\set{q_1,q_2,q_3}$ is not minimal, although it is an eventually configuration. 
Thus Algorithm \ref{newEventuallyAlg} returns the set of patterns $\set{\F q_0, \F (q_1 \vee q_3) , \F (q_2 \vee q_3)}$.

\end{example}


\subsubsection{Patterns of the Form $\Diamond \Box \psi$}

To compute formulas of the form $\Diamond \Box \psi$ which hold over all runs of the  
 FTS $\Tc=\seq{Q,\set{q_0},\delta}$ of the counter-strategy, we view $\Tc$ as a graph and separate 
its states into two groups: 
$Q^{cycle} \subseteq Q$, the set of states that are part of a cycle in $\Tc$ (including the cycle from one node 
 to itself), and $Q' = Q - Q^{cycle}$. 
Without loss of generality we assume that any state $q \in Q$ is reachable from $q_0$.
Therefore, any state $q \in Q^{cycle}$ belongs to a reachable strongly connected component $C$ of $\Tc$. 
Also for any strongly connected component $C$ of $\Tc$ , there exists a run $\sigma$ of $\Tc$ which reaches states in $C$ and keeps cycling there forever. 
Hence, the formula $\psi_1=\Diamond \Box \bigvee_{q \in C} q$ holds over the run $\sigma$. 
Indeed $\psi_1$ is the minimal formula of disjunctive form which holds over all runs that can reach the strongly connected component $C$. 
That is, by removing any of the states from $\psi_1$, one can find a run $\sigma'$ which can reach the strongly connected component $C$ and visit the removed state, 
falsifying the resulted formula. 
Therefore, eventually for any execution of $\Tc$, the state of the system will always be in one of the states $q \in Q^{cycle}$. 
Thus the formula $\psi = \Diamond \Box \bigvee_{q \in Q^{cycle}} q$ is the minimal formula of the form eventually always which holds over all runs of $\Tc$. 

To partition the states of the $\Tc$ into $Q^{cycle}$ and $Q'$ we use Tarjan's algorithm for computing strongly connected components of the graph. 
Thus the algorithm is of linear time complexity in number of states and transitions of $\Tc$. 

\begin{example}
 Consider the non-deterministic FTS shown in Figure \ref{exampleGraph}. It has three strongly connected components: $\set{q_0}$, $\set{q_1,q_2}$ 
and $\set{q_3}$. Only the latter two components include a cycle inside them, that is $Q^{cycle}=\set{q_1,q_2,q_3}$. 
Thus, the pattern $\psi=\Diamond \Box (q_1 \vee q_2 \vee q_3)$ is generated. 
Note that the possible runs of the system are $\sigma_1=q_0,(q_1,q_2)^\omega$ and $\sigma_2=q_0,(q_3)^\omega$. 
The generated pattern $\psi$ holds over both of these runs. 
Observe that removing any of the states in $\psi$ will result in a formula which is not satisfied by $\Tc$ any more.  
\end{example}

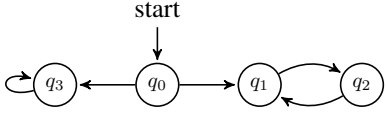
\begin{figure}[t]
\centering

 \begin{tikzpicture}[->,>=stealth',shorten >=1pt,auto,node distance=1.7cm,
                    semithick]
  \tikzstyle{state}=[circle, draw=black, scale=0.8]
  
  \node[state, initial, initial where=above] (A) {$q_0$};
  \node[state] (B) [right of=A] {$q_1$};
  \node[state] (C) [right of=B] {$q_2$};
  \node[state] (D) [left of=A] {$q_3$};

  \path (A) edge (B)
	    edge (D)
	(B) edge [bend left] (C)
	(C) edge [bend left] (B)
	(D) edge [loop left] (D);
 \end{tikzpicture}
\caption{A non-deterministic finite state transition system $\Tc$}
\label{exampleGraph} \vspace{-0.2in}
\end{figure} 

\subsubsection{Patterns of the Form $\Diamond (\psi_1 \wedge \bigcirc \psi_2)$}  
To generate candidates of the form $\Diamond (\psi_1 \wedge \bigcirc \psi_2)$, first note that $\Diamond (\psi_1 \wedge \bigcirc \psi_2)$ 
holds only if $\Diamond \psi_1$ holds. 
Therefore, a set of eventually patterns $\F \psi_1$ is first computed using Algorithm \ref{newEventuallyAlg}. 
Then for each formula $\Diamond \psi_1$, the pattern $\Diamond (\psi_1 \wedge \bigcirc \bigvee_{q \in Next(\psi_1)} q)$ is generated, where 
$Next(\psi_1)$ is the set of states that can be reached in one step from the configuration specified by $\psi_1$. 
Formally,  $Next(\psi_1)=\set{q_i \in Q ~|~ \exists q_j \in \mathcal{C} \text{ s.t. } (q_j,q_i) \in \delta}$ 
and $\mathcal{C}$ is the configuration 
represented by $\psi_1 = \bigvee_{q \in \mathcal{C}} q$. 
The most expensive part of this procedure is computing the eventually patterns, 
therefore its complexity is the same as Algorithm \ref{newEventuallyAlg}.
Due to the lack of space, the algorithms for computing $\Diamond \Box \psi$ and $\Diamond (\psi_1 \wedge \bigcirc \psi_2)$ patterns are given in the Appendix.

\begin{example}
 Consider the FTS shown in Figure \ref{exampleGraph}. 
Given the set of eventually formulas produced in Example \ref{eventuallyPatternEx}, 
patterns $\Diamond (q_0 \wedge \bigcirc (q_1 \vee q_3))$, $\Diamond ((q_1 \vee q_3) \wedge \bigcirc (q_2 \vee q_3))$ and 
$\Diamond ((q_2 \vee q_3) \wedge \bigcirc (q_1 \vee q_3))$ are generated.  
\end{example}

The procedures described for producing patterns, 
lead to assumptions which only include environment variables, and are enough for resolving unrealizability in our case studies. 
However, in general, GR(1) assumptions can also include the system variables. 
The procedures can  be easily extended to the general case (see the Appendix). 

The following theorem states that the procedures described in this section, generate the strongest patterns 
of the specified forms. Its proof can be found in the Appendix.
Removing the weaker patterns leads to shorter formulas which are easier for the user to understand. 
It also decreases the number of generated candidates at each step. 
More importantly, it leads to weaker assumptions on the environment that can be used to rule out the counter-strategy. 
If the restriction imposed by any of these candidates is not enough to make the specification realizable, 
the method analyzes the counter-strategy computed for the new specification
to find assumptions that can restrict the environment more. 
This way the counter-strategies guide the method to synthesize assumptions that can be used to achieve realizability.

\begin{theorem}
\label{strongestCandidates}
For any formula of the form $\F \psi, \F\G \psi$, or $\F(\psi_1 \wedge \X \psi_2)$ which hold over all runs of a given FTS $\Tc$, 
there is an equivalent or stronger formula of the same form synthesized by the algorithms described in Section \ref{patterns-synthesis}. 
\end{theorem}

\section{Case studies}
\label{sec:casestudies}
We now present two case studies. 
We use RATSY to generate counter-strategies and Cadence SMV model checker \cite{CadenceSMV} to check the consistency of the generated candidates. 
In our experiments, we set $\alpha$ in Algorithm \ref{findAssumptions} to two, 
and $\beta$ in Algorithm \ref{newEventuallyAlg} to the maximum outdegree of the vertices of the counter-strategy's abstract directed graph. 
We slightly change Algorithm \ref{findAssumptions} to find all possible refinements within the specified depth. 


\subsection{Lift Controller}
We borrow the lift controller example from \cite{bloem2012synthesis}. 
Consider a lift controller serving three floors. 
Assume that the lift has three buttons, denoted by the Boolean variables $b_1$, $b_2$ and $b_3$, which are controlled by the environment. 
The location of the lift is represented using Boolean variables $f_1$, $f_2$ and $f_3$ controlled by the system. 
The lift may be requested on each floor by pressing the corresponding button. 
We assume that $(1)$ once a request is made, it cannot be withdrawn, 
$(2)$ once the request is fulfilled it is removed,
and $(3)$ initially there are no requests. 
Formally, the specification of the environment is 
$\phi_e = \phi^e_{init} \wedge \phi^e_{1_1} \wedge \phi^e_{1_2} \wedge \phi^e_{1_3} \wedge \phi^e_{2_1} \wedge \phi^e_{2_2} \wedge \phi^e_{2_3}$, where  
$\phi^e_{init} = (\neg b_1 \wedge \neg b_2 \wedge \neg b_3 )$, 
$\phi^e_{1_i} = \Box(b_i \wedge f_i \rightarrow \bigcirc \neg b_i)$, and 
$\phi^e_{2_i} = \Box(b_i \wedge \neg f_i \rightarrow \bigcirc b_i)$ for  $1 \leq i \leq 3$. 

The lift initially starts on the first floor. We expect the lift to be only on one of the floors at each step. 
It can move at most one floor at each time step. We want the system to eventually fulfill all the requests. 
Formally the specification of the system is given as 
$\phi_s = \phi^s_{init} \wedge \phi^s_1 \bigwedge_i \phi^s_{2,i} \wedge \phi^s_3 \bigwedge_j \phi^s_{4,j} \wedge \phi^s_5$,  
where
\begin{itemize}
 \item $\phi^s_{init} = f_1 \wedge \neg f_2 \wedge \neg f_3$,
  \item $\phi^s_1 = \Box (\neg(f_1 \wedge f_2) \wedge \neg(f_2 \wedge f_3) \wedge \neg (f_1 \wedge f_3))$,
  \item $\phi^s_{2,i} = \Box (f_i \rightarrow \bigcirc (f_{i-1} \vee f_i \vee f_{i+1}))$,
  \item $\phi^s_3 = \Box((f_1 \wedge \bigcirc f_2) \vee (f_2 \wedge \bigcirc f_3) \rightarrow (b_1 \vee b_2 \vee b_3))$, and
  \item $\phi^s_{4,j}=\Box \Diamond (b_j \rightarrow f_j)$.
\end{itemize}

The requirement $\phi^s_3$ says that the lift moves up one floor only if some button is pressed. 
The specification $\phi = \phi_e \rightarrow \phi_s $ is realizable. Now assume that the designer wants to ensure that all floors are infinitely often visited; thus she 
adds the guarantees $\bigwedge_j \phi^s_{5,j}$ where $\phi^s_{5,j} = \Box \Diamond (f_j)$ to the set of system requirements. 
The specification $\phi' =  \phi_e \rightarrow \phi_s \bigwedge_j \phi^s_{5,j}$ is not realizable. 
A counter-strategy for the environment is to always keep all  $b_i$'s low. 
We run our algorithms with the set of all the environment variables $\set{b_1,b_2,b_3}$ for all assumption forms. 
The algorithm generates the refinements $\psi_1=\Box \Diamond (b_1 \vee b_2 \vee b_3)$
  and $\psi_2=\Box ((\neg b_1 \wedge \neg b_2 \wedge \neg b_3) \rightarrow \bigcirc (b_1 \vee b_2 \vee b_3))$. 
Refinement $\psi_1$ requires that the environment infinitely often presses a button. 
Refinement $\psi_2$ is another suggestion which requires the environment to make a request after any inactive turn. 
Refinement $\psi_1$ seems to be more reasonable and the user can add it to the specification to make it realizable. 

Only one counter-strategy is processed during the search for finding refinements
and three candidate assumptions are generated overall, where one of the candidates 
is inconsistent with $\phi'$ and the two others are refinements $\psi_1$ and $\psi_2$. 
Thus, the search terminates after checking the generated assumptions at first level. 
Only $0.6$ percent of total computatuion time was spent on generating candidate assumptions from the counter-strategy.
Note that to generate $\psi_1$ using the template-based method in \cite{LiDS11}, 
the user needs to specify a template with three variables which leads to $2^3=8$ candidate assumptions, 
although only one of them is satisfied by the counter-strategy.

\subsection{AMBA AHB}
ARM's Advanced Microcontroller Bus Architecture (AMBA) defines the Advanced High-Performance Bus (AHB) which is an on-chip communication protocol. 
Up to $16$ \emph{masters} and $16$ \emph{slaves} can be connected to the bus. 
The masters start the communication (read or write) with a slave and the slave responds to the request. 
Multiple masters can request the bus at the same time, but the bus can only 
be accessed by one master at a time. 
A bus access can be a single \emph{transfer} or a \emph{burst}, which consists of multiple number of transfers. 
A bus access can be locked, which means it cannot be interrupted.
Access to the bus is controlled by the \emph{arbiter}. 
More details of the protocol can be found in \cite{bloem2012synthesis}. 
We use the specification given by one of RATSY's example files (amba02.rat). 
There are four environment signals:

\begin{itemize}
 \item $\tt{HBUSREQ}$[$i$]: Master $i$ requests access to the bus. 
  \item $\tt{HLOCK}$[$i$]: Master $i$ requests a locked access to the bus. This signal is raised in combination with $\tt{HBUSREQ}$[$i$].
  \item $\tt{HBURST}$[$1:0$]: Type of transfer. Can be SINGLE (a single transfer), BURST4 (a four-transfer), or INCR (unspecified length burst).
  \item $\tt{HREADY}$: Raised if the slave has finished processing the data. The bus owner can change and transfers can start only when HREADY is high. 
\end{itemize}
The first three signals are controlled by the masters and the last one is controlled by the slaves.
The specification of amba02.rat consists of one master and two slaves. 
For our experiment, we remove the fairness assumption $\Box \Diamond \tt{HREADY}$ from the specification. 
The new specification is unrealizable. 
We run our algorithm with the sets of variables $\set{\tt{HREADY}}$, $\set{\tt{HREADY},\tt{HBUSREQ[0]} , \tt{HBUSREQ[1]} , \tt{HLOCK[0]} , \tt{HLOCK[1]}}$, $\set{\tt{HREADY}}$ and 
$\set{\tt{HBUSREQ[0] , HBUSREQ[1]}}$ to be used in liveness, safety, left and right hand side of transition assumptions, respectively. 
Some of the refinements generated by our method are:  $\psi_1 = \GF\tt{HREADY}$,  $\psi_2 = \G (\tt{HREADY}\vee\neg \tt{HBUSREQ[0]}\vee\neg \tt{HLOCK[0]}\vee\neg \tt{HBUSREQ[1]}\vee\neg \tt{HLOCK[1]})\wedge\GF \tt{HREADY}$, and 
$\psi_3 = 
	 \G (\tt{HREADY} \rightarrow \bigcirc \neg HBUSREQ[0]) \wedge \G (\neg \tt{HREADY} \rightarrow \bigcirc \neg HBUSREQ[0])$. 
	 Note that although $\psi_2$ is a consistent refinement, it includes $\psi_1$ as a subformula and it is more restrictive. 
	 The refinement $\psi_3$ implies 
that $\tt{HBUSREQ[0]}$ must always be low from the second step on. 
	 Among these suggested refinements, $\psi_1$ appears to be the best option. Our method only processes one counter-strategy with five states and generates five 
candidate assumptions to find the first refinement $\psi_1$. To find all refinements within the depth two, overall five counter-strategies are processed 
by our method during the search,  where the largest counter-strategy had $25$ states. 
The number of assumptions generated for each counter-strategy during the search is less than nine. 
$28.6$ percent of total computation time was spent on generating candidate assumptions from the counter-strategies.

\section{Conclusion and Future Work}
We presented a counter-strategy guided  approach for adding environment assumptions to an unrealizable specifications in order to achieve realizability. 
We gave algorithms for synthesizing weakest assumptions of certain forms (based on ``patterns'') that can be used to rule out the counter-strategy. 

We chose to apply explicit-state graph search algorithms on the
counter-strategy because the available tools for solving games output the counter-strategy as a graph in an explicit form.
Symbolic analysis of the counter-strategy may be  desirable
for scalability, but the key challenge for this is to develop algorithms for solving games that
can produce counter-examples in compact symbolic form. Synthesizing symbolic patterns is one of the future directions.

Counter-strategies provide useful information for explaining reasons for unrealizability. 
However, there can be multiple ways to rule out a counter-strategy. 
We plan to investigate how the multiplicity of the candidates generated  by our method can be used to synthesize better assumptions.  
Furthermore, our method asks the user for  subsets of variables to be used in generating candidates. 
The choice of the subsets can significantly impact how fast the algorithm can find a refinement. 
Automatically finding good subsets of variables that contribute to the unrealizability problem is another future direction. 
Synthesizing environment assumptions for more general settings, and using the method for synthesizing interfaces between components in context of 
compositional synthesis  are subject to our current work.

\label{sec:conclusion}


\bibliographystyle{plain}
\bibliography{papers}

\begin{thebibliography}{10}

\bibitem{Bloem_ratsy}
R.~Bloem, A.~Cimatti, K.~Greimel, G.~Hofferek, R.~K{\"o}nighofer, M.~Roveri,
  V.~Schuppan, and R.~Seeber.
\newblock Ratsy--a new requirements analysis tool with synthesis.
\newblock In {\em CAV 2010}, pages 425--429. Springer, 2010.

\bibitem{bloem2012synthesis}
R.~Bloem, B.~Jobstmann, N.~Piterman, A.~Pnueli, and Y.~Sa'ar.
\newblock Synthesis of reactive (1) designs.
\newblock {\em Journal of Computer and System Sciences}, 78(3):911--938, 2012.

\bibitem{ChatterjeeHJ08}
K.~Chatterjee, T.~Henzinger, and B.~Jobstmann.
\newblock Environment assumptions for synthesis.
\newblock In {\em CONCUR 2008}, pages 147--161. Springer, 2008.

\bibitem{debugCS}
R.~Konighofer, G.~Hofferek, and R.~Bloem.
\newblock Debugging formal specifications using simple counterstrategies.
\newblock In {\em FMCAD 2009}, pages 152--159, 2009.

\bibitem{kupferman2006safraless}
O.~Kupferman, N.~Piterman, and M.~Vardi.
\newblock Safraless compositional synthesis.
\newblock In {\em CAV 2006}, pages 31--44. Springer, 2006.

\bibitem{LiDS11}
W.~Li, L.~Dworkin, and S.~Seshia.
\newblock Mining assumptions for synthesis.
\newblock In {\em MEMOCODE 2011}, pages 43--50. IEEE, 2011.

\bibitem{CadenceSMV}
K.~McMillan.
\newblock Cadence {SMV}. http://www.kenmcmil.com/smv.html.

\bibitem{ozay2011distributed}
N.~Ozay, U.~Topcu, and R.~Murray.
\newblock Distributed power allocation for vehicle management systems.
\newblock In {\em CDC-ECC 2011}, pages 4841--4848. IEEE, 2011.

\bibitem{pitermangr1}
N.~Piterman, A.~Pnueli, and Y.~Sa'ar.
\newblock Synthesis of reactive (1) designs.
\newblock In {\em VMCAI 2006}, pages 364--380. Springer, 2006.

\bibitem{pnueli1989synthesis}
Amir Pnueli and Roni Rosner.
\newblock On the synthesis of a reactive module.
\newblock In {\em POPL 1989}, pages 179--190. ACM, 1989.

\bibitem{wongpiromsarn2010receding}
T.~Wongpiromsarn, U.~Topcu, and R.~M. Murray.
\newblock Receding horizon temporal logic planning.
\newblock {\em IEEE Transactions on Automatic Control}, 57(11):2817--2830,
  2012.

\end{thebibliography}
%

\clearpage
\section{Appendix}
\label{appendix}

\subsection{Simplifying the Generated Candidates}
\label{simple}
Some simple techniques are used to simplify the generated candidate assumptions for a given counter-strategy. 
We explain them over a synthesized liveness assumption $\GF \phi$. Other forms are simplified similarly. 
Note that $\phi$ is of conjunctive normal form, that is, $\phi=\bigwedge_j \bigvee_i \ell_i^j$ where $\ell^j_i$ is 
a literal over a Boolean variable $x_i$ in a clause $c_j$. 
The clauses correspond to the complement of the state predicates in our method. 
First, if a literal over a Boolean variable has the same form in all clauses, it is factored out from the clauses. 
For example, consider the formula $\phi=(a \vee b) \wedge (a \vee \neg b) \wedge (a \vee \neg b) \wedge (a \vee b)$, where 
$a$ and $b$ are Boolean variables. $a$ can be factored out giving $\phi=a \vee ((b) \wedge (\neg b) \wedge (\neg b) \wedge (b))$. 
We scan the formula and remove the repetitive clauses, for example $(\neg b)$ and $(b)$ clauses are repetitive in $\phi$, 
thus it can be simplified to $\phi=a \vee ((\neg b) \wedge (b))$. 
Finally, if there are two clauses with one variable, the formula can further be simplified as $\phi= a \vee False = a$. 
In future we plan to find better simplifying techniques for more general candidate assumptions. 
These simplifications is important because  one of our goals is to generate formulas which are easy for the user to understand.

\subsection{Algorithms}
Algorithm \ref{eventuallyAlwaysAlg} and Algorithm \ref{eventuallyNextAlg} generate the patterns of the forms 
$\F\G \psi$ and $\F(\psi_1 \wedge \bigcirc \psi_2)$, respectively.
\begin{algorithm}[h]
  \KwIn{Counter-strategy's abstract FTS $\Tc=\seq{Q,\set{q_0},\delta}$}
  \KwOut{A set of patterns of the form $\Diamond \Box \psi$ where $\Tc \models \Diamond \psi$}
  Let $Q^{cycle} =  \{q \in Q ~|~ \exists \text{ a cycle } \in \Tc \text{ including } q \}$\;
  return $\psi=\Diamond \Box \bigvee_{q \in Q^{cycle}} q$\;
  \caption{Generating $\Diamond \Box \psi$ patterns}
  \label{eventuallyAlwaysAlg}
\end{algorithm}

\begin{algorithm}[h]
  \KwIn{Counter-strategy's abstract FTS $\Tc=\seq{Q,\set{q_0},E}$}
\KwIn{$\beta$, maximum number of states in $\psi_1$ in generated patterns}
  \KwOut{a set of patterns of the form $\Diamond (\psi_1 \wedge \bigcirc \psi_2)$ where $\Tc \models \Diamond \psi$}
  $\F \text{Patterns}$ = patterns generated by Algorithm \ref{newEventuallyAlg} with input $\Tc$ and $\beta$\;
  Let Patterns = \textbf{Empty}\;
  \ForEach{formula $\Diamond \psi \in \F \text{Patterns}$}{
      Patterns = Patterns $\cup$ $\Diamond (\psi \wedge \bigcirc \bigvee_{q \in Next(\psi)} q)$\;
  }
  return Patterns;
  \caption{Generating $\Diamond (\psi_1 \wedge \bigcirc \psi_2)$ patterns}
  \label{eventuallyNextAlg}
\end{algorithm}

\subsection{Extending patterns to include system variables}
To be able to include system variables, we extend the finite state transition system with labels over transitions which are propositions over 
system variables, $O$. Formally an extended FTS is a tuple $\Tc^{ext}=\seq{Q,\set{q_0},\mathcal{L},\delta}$ where $Q,q_0$ and $\delta$ is similar to what we had before and 
$\mathcal{L}: \delta \rightarrow 2^O$ is a labeling function which maps each transition to a proposition over system variables. 
For each transition $e \in \delta$, $\mathcal{L}(e)=\bigwedge_i \ell_i$ where each $\ell_i$ is a literal over a variable $y_i \in O$.
Generated patterns are of the form 
$\F (\psi \wedge \bigvee_{e_i \in \text{ outgoing}(\psi)} \mathcal{L}(e_i)) , \F\G (\psi \wedge \bigvee_{e_i \in \text{ outgoing}(\psi)} \mathcal{L}(e_i))$ 
and $\Diamond ((\psi_1 \wedge \bigvee_{e_i \in \text{ outgoing}(\psi_1)} \mathcal{L}(e_i) \wedge \bigcirc \bigvee_{q_i \in Next(\psi)} q_i)$ 
where $outgoing(\psi)$ is the set of transitions going out of states included in $\psi$, i.e., transitions $e=(q_i,q_j) \in \delta$ such that $q_i \in \mathcal{C}_\psi$, and 
$\mathcal{C}_\psi$ is the set of states included in the formula $\psi$.

\subsection{Proof of Theorem \ref{strongestCandidates}}
Note that if $C$ is an eventually configuration, then any configuration $C'$ such that $C \subset C'$ is also an eventually configuration. 
Moreover,  $\F\bigvee_{q \in C} q \rightarrow \F\bigvee_{q' \in C'} q'$, that is, the formula corresponding to $C$ is stronger than the one corresponding to $C'$. 


We use the following lemma in proof of Theorem \ref{strongestCandidates}. 
Intuitively it says that any propositional formula $\phi$ over states $Q$ of $\Tc$ that hold over some run of it, can be written as 
disjunction of a subset of the states in $Q$. 

\begin{lemma}
\label{lemma:transform}
 Let $\Tc=\seq{Q,\set{q_0},\delta}$ be a finite transition system. 
Consider a propositional formula $\phi$ over the states in $Q$. Assume there exists a run $\sigma$ of $\Tc$ and $i \geq 0$ such that $\sigma_i \models \phi$. 
Then there exists $Q_{\phi'} \subseteq Q$ such that the formulas $\phi$ and $\phi'= \bigvee_{q \in Q_{\phi'}} q$ are equivalent. 
\end{lemma}

\begin{proof}
Without loss of generality assume that $\phi$ only includes negation and disjunction connectives. 
Note that for any run $\sigma$ of $\Tc$, $\sigma_i \models \neg \bigvee_{q \in Q'} q$ for some $Q' \subseteq Q$ and $i \geq 0$ 
if and only if $\sigma_i \models \bigvee_{q \in Q-Q'} q$. 
Therefore, subformulas of the form $\neg  \bigvee_{q \in Q'} q$ can be replaced by $\bigvee_{q \in Q-Q'} q$.  
Using this rule, all negation operators can be removed from $\phi$ to obtain an equivalent formula $\phi'$. 
Let $Q_{\phi'} \subseteq Q$ be the set of states $q \in Q$ which appears in $\phi'$. 
It follows that the formulas $\phi$ and $\phi'= \bigvee_{q \in Q_{\phi'}} q$ are equivalent.  
\end{proof}


\begin{proof}[Proof of Theorem \ref{strongestCandidates}]

We prove the theorem for Algorithms \ref{newEventuallyAlg} and \ref{eventuallyAlwaysAlg}.  
The proof for Algorithm \ref{eventuallyNextAlg} is similar.
First consider the eventually formulas $\F \psi$ generated by Algorithm \ref{newEventuallyAlg}. 
We assume that $\beta = 2^{|Q|}$, that is, the algorithm finds all minimal eventually configurations. 
Assume there exists a formula $\F \phi$ which holds over all runs of $\Tc$. 
By Lemma \ref{lemma:transform} there exists $Q_{\phi} \subseteq Q$ such that $\F \phi = \F (\bigvee_{q \in Q_{\phi}} q)$. 
Since $\F \phi$ holds over all runs of $\Tc$, $Q_{\phi}$ must be an eventually configuration. 
Algorithm \ref{newEventuallyAlg} finds all minimal eventually configurations of $\Tc$. 
Therefore, there exists a minimal eventually configuration $Q_{\psi} \subseteq Q$ corresponding to a 
formula $\F \psi$ generated by Algorithm \ref{newEventuallyAlg} such that $Q_{\psi} \subseteq Q_{\phi}$. 
It follows that $\F \psi \rightarrow \F \phi$. That is, there exists a formula generated by Algorithm \ref{newEventuallyAlg} which is stronger 
than or equivalent to $\F \phi$.

Eventually always formula $\F\G \psi$ generated by Algorithm \ref{eventuallyAlwaysAlg} is such that removing any state from $\psi$ makes the formula 
unsatisfiable and adding any state to it makes the formula weaker. 
Thus any formula $\F\G \phi$ which holds over all runs of $\Tc$ should be equivalent to or weaker 
than $\F\G \psi$.
\end{proof}

\end{document}